\documentclass[12pt]{iopart}

\usepackage{iopams, bbm,amsthm}

\newcommand{\herm}[1]{\mathrm{Herm}\left(#1\right)}
\newcommand{\lin}[1]{\mathrm{L}\left(#1\right)}
\newcommand{\unitary}[1]{\mathrm{U}\left(#1\right)}

\newcommand{\ip}[2]{\langle #1 , #2\rangle}
\newcommand{\norm}[1]{\| #1 \|}
\newcommand{\abs}[1]{| #1 |}
\newcommand{\op}[2]{|#1\rangle\langle#2|}

\newcommand\Hilbert{\mathcal H}

\newcommand\tF{\mathtt F}

\newcommand\fF{\mathfrak F}

\newcommand\E{\hat E}
\newcommand\F{\hat F}
\newcommand\A{\hat A}
\newcommand\B{\hat B}
\newcommand\M{\hat M}

\newcommand{\pr}{\mathrm{Pr}}

\newcommand{\id}{\mathbbm 1}

\newcommand\integer{\mathbb{Z}}

\newcommand\pure{\mathrm{pure}}

\newtheorem{theorem}{Theorem}
\newtheorem{definition}{Definition}

\begin{document}

\title[Frame representations]{Frame representations of quantum mechanics and the necessity of negativity in quasi-probability representations}
\author{Christopher Ferrie$^{1,2}$ and Joseph Emerson$^{1,2}$}
\address{$^1$ Department of Applied Mathematics, University of Waterloo, Waterloo, Ontario, Canada, N2L 3G1}
\address{$^2$ Institute for Quantum Computing, University of Waterloo, Waterloo, Ontario, Canada, N2L 3G1}
\eads{\mailto{csferrie@uwaterloo.ca} and \mailto{jemerson@uwaterloo.ca}}
\begin{abstract}
Several finite dimensional quasi-probability representations of
quantum states have been proposed to study various problems in
quantum information theory and quantum foundations.  These
representations are often defined only on restricted dimensions and
their physical significance in contexts such as drawing
quantum-classical comparisons is limited by the non-uniqueness of
the particular representation. Here we show how the mathematical
theory of frames provides a unified formalism which accommodates all
known quasi-probability representations of finite dimensional
quantum systems.  Moreover, we show that any quasi-probability
representation is equivalent to a frame representation and then
prove that any such representation of quantum mechanics must exhibit
either negativity or a deformed probability calculus.
\end{abstract}
\maketitle

\section{Introduction}

The Wigner function \cite{Wig32} is a
\emph{quasi-probability} density on a classical phase space which
represents a quantum state. The term quasi-probability refers to the fact
that the function is not a true density as it takes on negative values for some quantum
states. As is well known, the Wigner formalism can be lifted into a
fully autonomous phase space theory which reproduces all the
predictions of quantum mechanics \cite{Bak58}.

In recent years various phase space and other quasi-probability
representations of finite dimensional quantum systems have been
proposed.  In the remainder of the paper, whenever we refer to a representation of quantum systems, we implicitly mean a representation of \emph{finite-dimensional} quantum systems (of dimensional $d$).  For example, the Wootters \cite{Woo87} phase space
function is defined on an $d\times d$ lattice indexed by the
integers modulo $d$ for all $d$ dimensional Hilbert spaces where $d$
is prime - it is defined on the Cartesian product of these lattices
whenever $d$ is composite. Leonhardt \cite{Leo95} introduced a
four-fold redundant phase space function on a $2d\times 2d$ lattice
indexed by the integers modulo $2d$ which is valid for $d$ an even
integer. Heiss and Weigert \cite{HW00} have defined a phase space
function on $d^2$ points embedded in the sphere $\mathcal S^2$.
Gibbons, Hoffman, and Wootters \cite{GHW04} introduced a discrete
Wigner function on an $d\times d$ lattice indexed by elements of a
finite field of dimension $d$ when $d$ is a power of a prime number.
There are several others (for a recent review see \cite{Vou04}). In
addition to these discrete phase space functions, continuous phase
space representations of finite dimensional quantum state have also
been introduced \cite{VG72}, as well as more general
quasi-probability representations \cite{Har01,Hav03}, which are
real-valued representations that do not necessarily reflect any
preconceived classical phase space structure.

Such representations have provided insight into fundamental structures for finite-dimensional quantum systems.
For example, the representation proposed by Wootters identifies sets of
mutually unbiased bases \cite{Woo87,GHW04}. Inspired by the discovery that quantum resources
lead to algorithms that dramatically outperform their classical
counterparts, there has also been growing interest in the application of
discrete phase representation to analyze the quantum-classical
contrast for finite-dimensional systems, for example, quantum teleportation
\cite{Paz02}, the effect of decoherence on quantum walks
\cite{LP03}, quantum Fourier transform and Grover's algorithm
\cite{MPS02}, conditions for exponential quantum computational
speedup \cite{Gal05,CGG06}, and quantum expanders \cite{GJ07}.

A central concept in studies of the quantum-classical contrast in
the quasi-probability formalisms of quantum mechanics is the
appearance of \emph{negativity}.  A non-negative quasi-probability
function is a true probability distribution, prompting some authors
to suggest that the presence of negativity in this function is a
defining signature of non-classicality. However, a quantum state can
be negative in one representation and positive in another. This
simple fact underscores the obvious problem that considering any one
of these quasi-probability representations in the context of
determining criteria for the non-classicality of a given quantum
state is inadequate due to the non-uniqueness of that particular
representation. Ideally one would like to determine whether the
state can be expressed as a classical state in \emph{any}
quasi-probability representation. Indeed the sheer variety of
proposed quasi-probability representations prompts the question of
whether there is some shared underlying mathematical structure that
might provide a means for identifying the full family of such
representations.  The first goal of the work presented here is to
provide such a unifying formalism.

Moreover, from an operational point of view, states alone are an
incomplete description of an experimental arrangement.  For example,
Reference \cite{CGG06} proves that, within the class of
quasi-probability representations due to Gibbons et al \cite{GHW04},
the only positive pure states are a subset of the so-called
stabilizer states.  The authors note that these states are
``classical'' from the point of view of allowing an efficient
classical simulation via the stabilizer formalism. However, this set
of positive states includes the Bell states - states which
(maximally) violate a Bell inequality -  and hence these states are
maximally \emph{non-classical} according to a far more conventional
criterion of classicality: locality.

The resolution of this paradox is that one must also consider the
representation of measurements in the quasi-probability
representation in order to assess the classicality of a complete
\emph{experimental procedure}. Hence, it is important to elucidate
the ways in which a quasi-probability representation \emph{of states
alone} can be lifted to an autonomous quasi-probability
representation of \emph{both the states and measurements} defining a
set of complete experimental configurations. Indeed, in the
representation consider above, although the Bell state has a
positive representation, one can show that the conditional
probabilities representing the measurements assume negative values
(and hence are non-classical in a sense we make precise below).
Indeed the second goal of this work is to determine the full set of
possible quasi-probability representations of both states and
measurements. Only by considering this full set of quasi-probability
representations is it possible to establish a meaningful sense in
which the appearance of negativity provides a rigorous notion of
non-classicality, i.e., either the states or effects (or both) must
exhibit negativity in all such representations (otherwise a
classical representation does exist).

The outline of the paper is as follows.  In Section \ref{frames}, we
show how the mathematical theory of frames, which has been developed
in the context of signal analysis to devise methods of representing
information redundantly in order to protect it against noise
\cite{Chr03}, provides a formalism which underlies all known
quasi-probability representations of finite dimensional quantum
states. In Section \ref{quantum}, we show that there are two ways in
which any quasi-probability representation of states can be extended
to include a representation of measurements, and hence lifted to a
fully autonomous formulation of finite dimensional quantum
mechanics.  In Section \ref{negative}, we prove that any
representation that reproduces all of the predictions of quantum
mechanics must either i) exhibit negativity in the quasi-probability
functions for either states or measurements or ii) make use of a
deformed probability calculus, and then clarify in which sense these
correspond to non-classical properties. We conclude in Section
\ref{conclusion} by discussing how our formalism can be applied to
determine when non-classical resources are present in an
experimental system or a given quantum information task which
involve only a restricted set of preparations and measurements. In
the discussion we also connect our results with recent independent
work \cite{Spe07} which establishes negativity and contextuality as
equivalent criteria of the non-classicality of quantum mechanics.

\section{Frame Representations of Quantum States\label{frames}} From an operational point of view, the formulation of
quantum mechanics requires only the Hermitian operators acting on a
complex Hilbert space $\Hilbert$ with some finite dimension $d$ \cite{Har01}. The
Hermitian operators themselves form a real Hilbert space
$\herm\Hilbert$ of dimension $d^2$ with inner product
$\ip{\A}{\B}:=\tr(\A\B)$.

A basis is a linearly independent set that spans $\herm\Hilbert$.    A \emph{frame} is a
generalization of the notion of a basis.  Let $\Gamma$ be some set
with positive measure $\mu$.  The space of real valued square integrable functions on $\Gamma$ is denoted $L^2(\Gamma,\mu)$.
A \emph{frame} for $\herm\Hilbert$ is
a mapping $\F:\Gamma\to\herm\Hilbert$ which satisfies
\begin{equation}\label{def_frame}
a\norm{\A}^2\leq\int_\Gamma \rmd\mu(\alpha)\abs{\ip{\F(\alpha)}{\A}}^2\leq b\norm{\A}^2,
\end{equation}
for all $\A\in\herm\Hilbert$ and some constants $a,b>0$.  Note that (for finite-dimensional Hilbert spaces) a frame is equivalent to a spanning set which need not be linearly independent.  Such a linearly dependent spanning set is sometimes called an ``overcomplete basis''.
\begin{definition}
A mapping $\herm\Hilbert\to L^2(\Gamma,\mu)$ of the form
\begin{equation}\label{def_frame_rep}
\A\mapsto A(\alpha):=\ip{\F(\alpha)}{\A},
\end{equation}
where $\F$ is a frame, is a \emph{frame representation} of $\herm\Hilbert$.
\end{definition}
A \emph{dual frame} is a frame $\E:\Gamma\to\herm\Hilbert$ which satisfies
\begin{equation}\label{def_dual}
\A=\int_\Gamma \rmd\mu(\alpha)\ip{\F(\alpha)}{\A}\E(\alpha),
\end{equation}
for all $\A\in\herm\Hilbert$.  When $\Gamma$ is finite and
$\abs \Gamma=d^2$, the dual frame is the unique, otherwise there are infinitely many choices for a
dual frame.

Here we present two examples of frames.  Consider the operator $\hat Z$ whose spectrum is $\mathrm{spec} (\hat
Z)=\{\rme^ { \frac{2\rmi\pi}{d}k}: k\in \mathbb Z_d\}$.  The
eigenvectors form a basis for $\Hilbert$ and are denoted
$\{\phi_k:k\in\mathbb Z_d\}$. Consider also the operator defined by
$\hat X\phi_k=\phi_{k+1}$, where all arithmetic is modulo $d$.  The
operators $\hat Z$ and $\hat X$ are often called \emph{generalized
Pauli operators} since they are indeed the usual Pauli operators
when $d=2$.  The \emph{parity operator} is defined by $\hat P\phi_k=\phi_{-k}$.
Let $\Gamma=\mathbb Z_d\times\mathbb Z_d$ and $\mu$ be the counting measure and suppose $d$ is prime.  Consider the map $\F:\Gamma\to \herm\Hilbert$ defined by
\begin{equation}
\F(q,p)=\frac{1}{d^2}\hat X^{2q}\hat Z^{2p}\hat P \rme^{\frac{4\rmi \pi}{d}qp}.\label{Woottersframe}
\end{equation}
This map is a frame for $\herm\Hilbert$.  It also is an orthogonal basis for $\herm\Hilbert$ and thus the dual frame is unique.  Now let $\Gamma=\mathbb Z_{2d}\times\mathbb Z_{2d}$ and $\mu$ be the counting measure and suppose $d$ is even.  Consider the map $\F:\Gamma\to \herm\Hilbert$ defined by
\begin{equation}
\F(q,p)=\frac{1}{4d^2}\hat X^{q}\hat Z^{p} \hat P \rme^{\frac{\rmi \pi}{d}qp}.\label{evenframe}
\end{equation}
This map is also a frame.  However $\abs \Gamma=4d^2$ and thus the dual frame cannot be unique.

We propose the following as a minimal requirement for the definition of a quasi-probability representation of quantum \emph{states}.
\begin{definition}
A \emph{quasi-probability representation of quantum states} is any map
$\herm\Hilbert\to L^2(\Gamma,\mu)$ that is linear and invertible.
\end{definition}
Given this definition, any phase space representation is
then a particular type of quasi-probability representation.  In particular, if there exists symmetry group on $\Gamma$, $G$, carrying a
unitary representation $\hat U:G\to\unitary\Hilbert$ and  a
quasi-probability representation satisfying the covariance property
$\hat U_g\A\hat U_g^\dag\mapsto \{A(g(\alpha))\}_{\alpha\in\Gamma}$ for
all $\A\in\herm\Hilbert$ and $g\in G$, then $\hat A\mapsto A(\alpha)$ is a \emph{phase space
representation}.
All phase space functions (that we are aware of) in the literature
correspond to quasi-probability representations that satisfy this additional covariance condition.

It is clear that a frame representation defined by \Eref{def_frame_rep} is a linear bijection and hence a quasi-probability representation.  Thus the frames defined by \Eref{Woottersframe} and \Eref{evenframe} are quasi-probability representations.  Indeed, \Eref{Woottersframe} defines the phase space quasi-probability function defined by Wootters \cite{Woo87} while \Eref{evenframe} defines the phase space quasi-probability function defined by Leonhardt \cite{Leo95}.  It is less obvious that the converse is also true.  Nevertheless, the following theorem verifies this fact.
\begin{theorem}
If A mapping $W$ is a quasi-probability representation, then it is a frame representation for a unique frame $\F$.
\end{theorem}
\begin{proof}
Linearity and the Riesz
representation theorem implies that
$W(\A)(\alpha)=\ip{\F(\alpha)}{\A}$ for some unique mapping
$\F:\Gamma\to\herm \Hilbert$ (not necessarily a frame).  Since $\herm\Hilbert$ is finite dimensional, the inverse $W^{-1}$ is bounded.  Thus $W$ is bounded below by the bounded inverse theorem. That is, there exists a constant $a>0$ such that
\begin{equation*}
a\norm{\A}^2\leq\int_\Gamma \rmd\mu(\alpha)\abs{\ip{\F(\alpha)}{\A}}^2.
\end{equation*}
Since $\ip{\F(\alpha)}{\A}\in L^2(\Gamma,\mu)$, there exists a constant $b>0$ such that
\begin{equation*}
\int_\Gamma \rmd\mu(\alpha)\abs{\ip{\F(\alpha)}{\A}}^2\leq b\norm{\A}^2.
\end{equation*}
Hence $\F$ is a frame.
\end{proof}

\section{Frame Representations of Quantum Mechanics\label{quantum}} Most proposed
phase space functions (of finite-dimensional quantum systems) are representations of
quantum states alone.  Here we show that there are two approaches within the frame
formalism to lift any representation of states to a fully autonomous
representation of finite dimensional quantum mechanics.

An operational set of axioms \cite{Har01} of quantum mechanics are
\begin{enumerate}
\item[(i)] There exists a Hilbert space $\Hilbert$, $\dim\Hilbert =d$.
\item[(ii)] A preparation (state) is represented by a density operator $\hat\rho$ satisfying
$\ip{\psi}{\hat\rho\psi}\geq0$, for all $\psi\in\Hilbert$, and $\tr(\hat\rho)=1$.
\item[(iii)] A measurement is represented by a set of effects $\{\M_k\}$, i.e., positive operator valued
measure (POVM),  satisfying $0\leq\ip{\psi}{\M_k\psi}\leq1$, for all
$\psi\in\Hilbert$, and $\sum_k\M_k=\hat\id$.
\item[(iv)] For a system with density operator $\hat\rho$ subject to the measurement $\{\M_k\}$, the probability of obtaining outcome $k$ is given by the Born rule
\begin{equation}
\pr(k)=\tr(\M_k\hat\rho).
\end{equation}
\end{enumerate}


Hence to construct an autonomous formulation of quantum mechanics we
need a set of functions $\{M_k\}$ on phase space representing the
set of measurement operators $\{\M_k\}$, as well as a prescription
for calculating the probabilities that are prescribed by the Born
rule.

\subsection{Deformed probability representations}
The first frame representation approach to an autonomous formulation
of quantum mechanics consists of mapping both states and
measurements to $L^2(\Gamma, \mu)$ via the same frame $\F$, i.e.
$\hat\rho\mapsto \rho(\alpha):=\ip{\hat\rho}{\F(\alpha)}$ and
$\M_i\mapsto M_i(\alpha):=\ip{\M_i}{\F(\alpha)}$.  The functions in
the range of this frame representation, when the domain is
restricted to the density operators, are called
\emph{quasi-probability densities}. Similarly, the functions in the
range of the frame representation, when the domain is restricted to
the effects, are called \emph{conditional quasi-probabilities}. Then
the axioms of quantum mechanics become
\begin{enumerate}
\item There is a measurable set of allowed properties ${\Gamma}$ endowed with a positive measure $\mu$.
\item A preparation (state)
is represented by a quasi-probability density
$\rho(\alpha)\in\mathbb R$ which satisfies the normalization
condition $\int_{\Gamma} d\mu( \alpha) \rho( \alpha)=1$.
\item A measurement is represented by a set of conditional quasi-probabilities
$\{M_k(\alpha)\in\mathbb R\}$ which satisfies $\sum_k
M_k( \alpha)=1$ for all $ \alpha\in {\Gamma}$.
\item For a system with
quasi-probability density $\rho$ subject to the measurement $\{M_k\}$, the probability of obtaining outcome $k$ is given by
\begin{equation}\label{deformed}
\pr(k)=\int_{\Gamma} d\mu( \alpha,\beta) \rho( \alpha) M_k( \beta)\ip{\E(\alpha)}{\E(\beta)},
\end{equation}
where $\E$ is any frame dual to $\F$.
\end{enumerate}

As will become clear in the next section, \Eref{deformed} is a
\emph{deformed} version of the usual law of total probability and
hence we call this first approach a \emph{deformed probability
representation of quantum mechanics}.  This is not the only
possibility.  Indeed, below we will see a different approach.

\subsection{Quasi-probability representations}
Notice that the deformed probability calculus \Eref{deformed} can be written
\begin{equation}
\pr(k)=\int_{\Gamma}\rmd\mu(\alpha)\;{\rho(\alpha)}\tilde M_k(\alpha),\label{lawtotalprob}
\end{equation}
where
\begin{equation}
\tilde M_k(\alpha)=\int_{\Gamma}\rmd\mu(\beta)\;M_k(\beta)\ip{\E(\alpha)}{\E(\beta)}.\label{M}
\end{equation}
Recall that $M_k$ is the frame representation of $\M_k$ for the
frame $\F$.  Hence $\tilde M_k$ can be identified as the frame
representation of $\M_k$ using a frame $\E$ that is dual to $\F$.
The second frame representation approach to an autonomous
formulation consists of mapping density operators to functions in
$L^2(\Gamma, \mu)$ via a particular choice of frame $\F$ and effects
to functions in $L^2(\Gamma,\mu)$ via a frame $\E$ that is dual to
$\F$, i.e. $\hat\rho\mapsto \rho(\alpha):=\ip{\hat\rho}{\F(\alpha)}$
and $\M_k\mapsto M_k(\alpha):=\ip{\M_k}{\E(\alpha)}$.  As above we
define the former functions to be quasi-probability densities and
the latter functions to be conditional quasi-probabilities.  The
axioms of quantum mechanics can be reformulated once again as
\begin{enumerate}
\item There is a set of allowed properties ${\Gamma}$ with a positive measure $\mu$.
\item A preparation (state)
is represented by a quasi-probability density
$\rho(\alpha)\in\mathbb R$ which satisfies the normalization
condition $\int_{\Gamma} d\mu( \alpha) \rho( \alpha)=1$.
\item A measurement is represented by a set of conditional quasi-probabilities
$\{\tilde M_k(\alpha)\in\mathbb R\}$ which satisfies $\sum_k
\tilde M_k( \alpha)=1$ for all $ s\in {\Gamma}$.
\item For a system with
probability density $\rho$ subject to the measurement $\{\tilde M_k\}$, the probability of obtaining outcome $k$ is given by \Eref{lawtotalprob}.
\end{enumerate}

As will become clear in the nexr section, the probability calculus
\Eref{lawtotalprob} given in condition (iv) is now just the usual
law of total probability, although the preparation and measurement
functions are not necessarily positive semi-definite.  Hence we call
this second approach a \emph{quasi-probability representations of
quantum mechanics} (i.e., a quasi-probability representation of both
states \emph{and} measurements).

Note that for quasi-probability representation the frame and its dual are required.  Recall that the frames given in Equations \eref{Woottersframe} and \eref{evenframe} defined the phase space quasi-probability functions of Wootters and Leonhardt.  For the Wootters case, the dual frame is unique and is given by
\begin{equation*}
\E(q,p)=\frac{1}{d}\hat X^{2q}\hat Z^{2p}\hat P \rme^{\frac{4\rmi \pi}{d}qp}.\label{Woottersdualframe}
\end{equation*}
For the Leonhardt case, the frame in \Eref{evenframe} is not a basis and the dual is not unique.  However, a quasi-probability representation of quantum mechanics only require \emph{a} dual frame.  One such dual frame is
\begin{equation*}
\E(q,p)=\frac{1}{2d}\hat X^{q}\hat Z^{p} \hat P \rme^{\frac{\rmi \pi}{d}qp}.\label{evendualframe}
\end{equation*}

\section{Non-classicality: negative quasi-probability or a deformed law of total probability\label{negative}}
Let the set $\Gamma$ represent the properties of a classical system
and the function $\rho(\alpha)>0$ represent the probabilistic
knowledge of these properties.   Note that these probability
densities form a convex set with the Dirac measures as its extreme
points.  A measurement is a partitioning of the space $\Gamma$ into
disjoint subsets $\{\triangle_j\}$.  The probability of the system
to have properties in $\triangle_j$ (we will call this ``outcome
$j$'') is
\begin{equation*}
\pr(j)=\int_{\triangle_j}\rmd\mu(\alpha)\;\rho(\alpha)=\int_\Gamma \rmd\mu(\alpha) \chi_{j}(\alpha)\rho(\alpha),
\end{equation*}
where $\chi_j(\alpha)\in\{0,1\}$ is the indicator function of
$\triangle_j$.  The measurement is equivalently specified by the set
$\{\chi_j(\alpha)\}$, which is interpreted as the conditional
probability of outcome $k$ given the systems is known to have the
properties $\alpha$.  A measurement of this type is deterministic:
it reveals with certainty the properties of the system.  Consider
now an indeterministic measurement specified by the conditional
probabilities $\{M_k(\alpha)\in[0,1]\}$.  These can always be
thought of as a convex combination of indicator functions.  That is,
the measurement functions form a convex set with the indicator
functions as its extreme points.  We summarize the above description
with the following definition.
\begin{definition}
Any statistical or operational model of a set of experimental
configurations is  \emph{classical} if all of the following
properties hold:
\begin{enumerate}
\item There is a set of allowed properties $\Gamma$ with a positive measure $\mu$.
\item A preparation (state)
is represented by a probability density $\rho(\alpha) \geq 0$ which
satisfies the normalization condition $\int_\Gamma d\mu(\alpha)
\rho(\alpha)=1$.
\item A measurement is represented by a set
$\{M_k(\alpha)\in[0,1]\}$ which satisfies $\sum_k
M_k(\alpha)=1$ for all $\alpha\in \Gamma$.
\item For a system with
probability density $\rho$ subject to the measurement $\{M_k\}$, the
probability of obtaining outcome $k$ is given by the law of total
probability
\begin{equation}\label{classical_prob_def}
\pr(k)=\int_\Gamma \rmd\mu(\alpha) \rho(\alpha) M_k(\alpha).
\end{equation}
\end{enumerate}
\end{definition}

Consider now a frame representation defined via a \emph{positive}
frame $\F$. Applying the deformed probability representation (the
first approach of the previous section) to map quantum mechanics to
the space of functions $L^2(\Gamma,\mu)$, we find that the
representations of the preparations and measurements satisfy the
criteria of the \emph{classical} model because they are guaranteed
to be non-negative function when the frame  $\F$ is positive.
However, as noted previously, the calculation of probabilities
\Eref{deformed} is \emph{deformed} when compared to the classical
one \Eref{classical_prob_def}. Hence under this approach the
associated frame representations do not meet the criteria of a
classical model.

Now, applying instead the quasi-probability representation (the
second approach of the previous section) the probability calculus
\Eref{lawtotalprob} is the same as the classical one
\Eref{classical_prob_def}. Furthermore the preparations are
represented by non-negative functions (because the frame  $\F$ is
positive) and therefore also meet the criteria set out by condition
(ii). However, in this case the measurements must be represented via
a frame $\E$ which is dual to the frame $\F$ that is used for
representing preparations. It is not immediately obvious that any
quasi-probability representation of quantum mechanics following this
second approach is also unable to meet the criteria of a classical
operational model, in particular, condition (iii) which requires
non-negative conditional probabilities. We now show that this is
impossible by proving that there does not exist a dual frame of
positive operators for a frame of positive operators.

\begin{theorem}
There does not exist a dual frame of positive operators for a frame of positive operators.
\end{theorem}
\begin{proof}
Consider
the mapping
\begin{equation}\label{entanglementbreaking}
\tilde\Phi(\A)=\int_\Gamma \rmd\mu(\alpha)\ip{\F(\alpha)}{\A}\E(\alpha),
\end{equation}
If $\tilde \Phi$ were the identity super-operator, then by
definition $\E$ would be the dual frame of $\F$.  We will show
this is not possible when both $\F$ and $\E$ are positive frames.
Let $\{\op{\phi_i}{\phi_j}: i,j\in\integer_d\}$ be the standard basis
for $\lin\Hilbert$.  Then the Choi-Jamiolkowski \cite{Jam72} representation of
$\tilde\Phi$ is
\begin{eqnarray}
J(\tilde\Phi)&=\sum_{i,j\in\integer_d}\tilde\Phi(\op{\phi_i}{\phi_j})\otimes \op{\phi_i}{\phi_j}\\
&=\int_\Gamma \rmd\mu(\alpha)\left(\sum_{i,j\in\integer_d}\langle{\phi_j}|{\F(\alpha)|\phi_i}\rangle\E(\alpha)\otimes \op{\phi_i}{\phi_j}\right)\\
&=\int_\Gamma \rmd\mu(\alpha)\left(\E(\alpha)\otimes \sum_{i,j\in\integer_d}\langle{\phi_j}|{\F(\alpha)|\phi_i}\rangle\op{\phi_i}{\phi_j}\right)\\
&=\int_\Gamma \rmd\mu(\alpha)\left(\E(\alpha)\otimes\F(\alpha)\right),
\end{eqnarray}
which is a separable operator (a convex combination of positive operators of the form $\A\otimes\B$) on
$\Hilbert\otimes\Hilbert$ when both $\F$ and $\E$ are positive
frames. However, $J(\tilde\id)$ is not a separable operator on
$\Hilbert\otimes\Hilbert$ and thus $\tilde\Phi$ cannot be the
identity super-operator.  Hence $\E$ cannot be a dual frame of
$\F$.
\end{proof}
This theorem can also be proven using the results of Reference \cite{HSR03}.
Theorem 2 of that paper shows that the channel $\tilde\Phi$ defined
by \Eref{entanglementbreaking} for positive operators $\F$ and $\E$
is so-called \emph{entanglement breaking}.  However, Theorem 6 of Reference
\cite{HSR03} states that if $\tilde\Phi$ has fewer than $d$ Kraus
operators, it is \emph{not} entanglement breaking.  Since the
identity superoperator has fewer than $d$ Kraus operators,
$\tilde\Phi$ is not entanglement breaking and therefore $\E$ is not
the dual of $\F$.

Hence, although quantum states can always be represented as
non-negative probabilities, measurement functions must then take on
negative values, or vice-versa.  In this way we have a direct proof
that there does not exist any choice of quasi-probability
representation of quantum mechanics that can be made consistent with
the non-negativity conditions associated with a classical model of
statistical events.

\section{Discussion\label{conclusion}}

Our results prove that the full spectrum of experimental statistics
prescribed by finite dimensional quantum theory can not be described
by any classical model consisting of the usual rules of probability
applied over an arbitrary choice of property (or hidden variable)
space. Equivalently stated, there does not exist a space of events
upon which one can formulate a non-negative quasi-probability
representation of quantum mechanics.

A promising application of the formalism we have developed is to
address the question of whether a restricted set of preparations and
measurements involves non-classical resources. This question has
arisen in the context of the degree of coherent control over quantum
systems, for example, in experiments involving nuclear magnetic
resonance or super-conducting devices, where the quantum states and
effects that can be achieved are restricted due to thermalization
and decoherence. Our formalism leads to a broadly enabling and
rigorous approach to determining the extent to which quantum effects
are indeed present in those systems. Another context in which this
questions arises is the field of quantum information and
computation. One has a task which can be achieved with a restricted
set of quantum preparations and effects and one would like to know
whether non-classical resources are actually required for that task.
In both of these contexts, if we can identify a particular frame and
a dual which can represent the restricted set of states and
measurements as non-negative functions then we can show that the
task or process can be represented as a classical statistical
process, and hence prove that it does not require quantum resources.
Conversely, if one can prove that no such choice of frames exists,
then one can prove that quantum resources are indeed necessary.

Finally, we conclude by addressing the question of how the notion of
non-classicality established by the absence of non-negative
quasi-probability representation relates to another fundamental
notion of non-classicality in quantum mechanics, namely,
\emph{contextuality}. The traditional definition of contextuality
comes from a theorem due to Kochen and Specker \cite{KS67}. The
Kochen-Specker theorem establishes a contradiction between a set of
plausible assumptions associated with the idea that quantum systems
possess pre-existing values for the outcomes of measurements, as is
the case in the classical world. Assuming that physical systems do
possess pre-existing values that are revealed via measurements, the
Kochen-Specker theorem leads to the following counterintuitive
property that such pre-existing values must satisfy \cite{Ish95}:
suppose three operators $\A$, $\B$, and $\hat C$ satisfy
$[\A,\B]=0=[\A,\hat C]$, but $[\B,\hat C]\neq0$, then the
pre-existing value of the observable $A$ will depend on whether
observable $B$ or $C$ is will be measured along with $A$. That is,
the pre-existing value of $A$ depends on the \emph{context} of the
measurement. We note that the notion of context independence is at
the heart of Bell-type inequalities, where the pre-existing values
of the commuting operators in question are required to be context
independent by appealing to local causality.

Spekkens has generalized the notion of \emph{non-contextuality} from
the idea that outcomes of individual measurements are independent of
the measurement context to the requirement that the
\emph{probabilities} for outcomes of measurements are independent of
the measurement context \cite{Spe05}.  This is achieved by
formulating a definition of contextuality for an arbitrary
operational theory and includes a notion of contextuality for
preparation procedures (states) as well as measurements. In
Reference~\cite{Spe07}, Spekkens has shown that a quasi-probability
representation of quantum mechanics which excludes negativity is
equivalent to the generalized notion of non-contextuality that he
proposed in Reference \cite{Spe05} and has obtained an independent
proof the impossibility of constructing a non-negative
quasi-probability representation. Interestingly, in light of this
connection our direct proof of the non-existence of a positive dual
frame to a positive frame gives a new independent proof of the
generalized contextuality of quantum mechanics.

In this paper we have shown that using frame theory provides a
formalism that unifies the known quasi-probability representations
of quantum \emph{states}.  We have shown two different ways (the
\emph{deformed} and \emph{quasi-probability} approach) to lift a
quasi-probability representation of states to a consistent and
equivalent formulation of quantum mechanics.  We have also proved
that these quasi-probability representations of quantum states and
\emph{measurements} require either negativity or a deformation of
the rule for calculating probabilities.  We have thus given a
mathematically rigorous set of criteria that establish the (long
suspected) connection between negativity and non-classicality. While
the results of this paper have been proven only for finite
dimensional Hilbert spaces (although allowing for either finite or
continuous representation spaces), we conjecture that the results
continue to hold also for infinite dimensional quantum systems
(i.e., all separable Hilbert spaces).

\ack The authors thank John Watrous for the simplified proof of the
negativity result and Bernard Bodmann, Matt Leifer, Etera Livine,
and Rob Spekkens for helpful discussions. This work was supported by
NSERC and MITACS.

\appendix

\section*{Appendix}

For completeness we describe here how to formulate quantum mechanics
directly in either the deformed and quasi-probability
representations without appealing to the axioms of quantum mechanics
in their usual formulation (i.e., in terms of positive linear
operators on finite-dimensional Hilbert space). Recall that
previously, in the deformed probability representation, a
quasi-probability density was defined as a function in the range of
a frame representation when the domain is restricted to the density
operators.  Similarly the conditional quasi-probabilities were those
functions in the range of a frame representation when the domain is
restricted to the effects. Of course, for a particular choice of
frame, not every function in $L^2(\Gamma,\mu)$ will correspond to a
valid quantum state or effect.  Hence we need a set of
\emph{internal} conditions, without appealing to the nature of the
linear operators in the standard formulation of quantum theory,
which characterize the valid state and measurement functions in
$L^2(\Gamma,\mu)$.  The conditions can be found by simply noting
that the frame representation \Eref{def_frame_rep} is an isometric
and algebraic isomorphism from $\herm\Hilbert$ to $L^2(\Gamma,\mu)$
equipped with inner product
\begin{equation*}
\ip{A}{B}_\tF:=\int_{\Gamma^2}\rmd\mu(\alpha,\beta)\;{A(\alpha)}B(\beta)\tF(\alpha,\beta),\label{def_ipGamma}
\end{equation*}
where $\tF(\alpha,\beta):=\ip{\E(\alpha)}{\E(\beta)}$, and algebraic multiplication
\begin{equation*}\label{def_starprod}
(A\star_\fF B)(\alpha):=\int_{\Gamma^2}\rmd\mu(\beta,\gamma)\;A(\beta)
B(\gamma) \fF(\alpha,\beta,\gamma),
\end{equation*}
where $\fF(\alpha,\beta,\gamma)=\ip{\F(\alpha)}{\E(\beta)
\E(\gamma)}$.

Using the above, we first state the conditions for a function in
$L^2(\Gamma,\mu)$ to be a valid state or effect in the deformed
probability representation. A \emph{pure state} is a function
$\rho_\pure\in L^2(\Gamma,\mu)$ satisfying $\rho_\pure\star_\fF
\rho_\pure=\rho_\pure$.  A general \emph{state} is a function
$\rho\in L^2(\Gamma,\mu)$ satisfying
$\ip{\rho}{\rho_\pure}_\tF\geq0$ for all pure states and
$\int_\Gamma d\mu(\alpha)\rho(\alpha)=1$.  A \emph{measurement} is
represented by a set $\{M_k\in L^2(\Gamma,\mu)\}$ of \emph{effects}
which satisfies $\ip{M_k}{\rho_\pure}_\tF\geq0$ for all pure states
and for which $\sum_{k}M_k=\id,$ where $\id$ is the identity element
in $L^2(\Gamma,\mu)$ with respect to the algebra defined by
$\star_\fF$.

For quasi-probability representations of quantum mechanics, the term
quasi-probability density has the same meaning as in the deformed
probability representation.  Similarly, the conditional
quasi-probabilities are those functions in the range of the frame
representation of the measurements (i.e. the frame representation
defined via the dual $\E$) when the domain is restricted to the
effects.  In this representation states and measurements in
$L^2(\Gamma, \mu)$ must again meet certain criteria to be valid. The
conditions are similar to those in the deformed probability
representation.  In particular, the pure states and general states
are equivalently characterized. However, a measurement is now
represented by a set $\{M_k\in L^2(\Gamma,\mu)\}$ which satisfies
$\ip{M_k}{\rho_\pure}\geq0$ (now the usual pointwise inner product)
for all pure states and for which $\sum_{k}M_k=\id,$ where $\id$ is
the identity element in $L^2(\Gamma,\mu)$ with respect to the
algebra defined by $\star_{\mathfrak E}$ (which is defined in the
same way as $\star_\fF$ with the roles of the frame and its dual
reversed).

\bibliographystyle{unsrt}
\section*{References}
\bibliography{mybib}

\end{document}